\documentclass[fullpage,times,11pt]{article}

\usepackage{times}
\usepackage{epsfig}
\usepackage{amsfonts}
\usepackage{amssymb}
\usepackage{amstext}
\usepackage{amsmath}
\usepackage{xspace}
\usepackage{theorem}
\usepackage{graphicx}
\usepackage{graphics}
\usepackage{colordvi}

\usepackage{fullpage}

\usepackage{marginnote}

\renewcommand{\Pr}{{\mathbf {Pr}}}
\newcommand{\channel}{{\rm channel}}
\newcommand{\from}{{\leftarrow}}
\newcommand{\decerr}{{\Pr_{\rm dec}}}

\newcommand{\E}{{\mathbf E}}
\newcommand{\calE}{{\cal E}}

\newcommand{\decerrT}{{\Pr_{\rm dec, T}}}

\newcommand{\calc}{{\cal C}}
\newcommand{\cald}{{\cal D}}
\newcommand{\calx}{{\cal X}}
\newcommand{\caly}{{\cal Y}}
\newcommand{\calp}{{\cal P}}
\newcommand{\sinf}{S_{\infty}}
\newcommand{\hinf}{h_{\infty}}
\newcommand{\pinf}{p_{\infty}}

\newcommand{\eps}{\varepsilon}

\newcommand{\R}{{\mathbb R}}
\newcommand{\Z}{{\mathbb Z}}
\newcommand{\delay}{\Delta}
\newcommand{\noise}{\xi}
\newcommand{\mcap}{{\sf Cap}}

\newtheorem{lemma}{Lemma}[section]
\newtheorem{theorem}[lemma]{Theorem}
\newtheorem{observation}[lemma]{Observation}
\newtheorem{corollary}[lemma]{Corollary}
\newtheorem{claim}[lemma]{Claim}
\newtheorem{proposition}[lemma]{Proposition}
\newtheorem{assumption}[lemma]{Assumption}
\newtheorem{definition}[lemma]{Definition}

\newenvironment{proof}{{\bf Proof:}}{\hfill\stopproof}
\def\stopproof{\square}
\def\square{\vbox{\hrule height.2pt\hbox{\vrule width.2pt height5pt \kern5pt
\vrule width.2pt} \hrule height.2pt}}

\begin{document}

\title{Delays and the Capacity of Continuous-time Channels}

\author{
Sanjeev Khanna\thanks{
Department of Computer and Information Science,
University of Pennsylvania, Philadelphia PA. Email: {\tt
sanjeev@cis.upenn.edu}. Supported in part by NSF Awards CCF-0635084 and IIS-0904314.}
\and Madhu Sudan\thanks{Microsoft Research, Cambridge, MA 02142, and MIT, Cambridge, MA 02139.
Email: {\tt madhu@mit.edu}. }}

\maketitle

\begin{abstract}
Any physical channel of communication offers two potential reasons
why its capacity (the number of bits it can transmit in a unit of time) might
be unbounded: (1) (Uncountably) infinitely many choices of signal strength at
any given instant of time, and (2) (Uncountably) infinitely many instances of time at
which signals may be sent. However channel noise cancels out the
potential unboundedness of the first aspect, leaving typical channels
with only a finite capacity per instant of time. The latter source
of infinity seems less extensively studied. A potential source
of unreliability that might restrict the capacity also from the
second aspect is ``delay'': Signals transmitted by the sender at
a given point of time may not be received with a predictable delay at the
receiving end. In this work we examine this source of uncertainty by
considering a simple discrete model of delay errors.
In our model
the communicating parties get to subdivide time as microscopically finely as they wish,
but still have to cope with communication delays that are macroscopic and variable.
The continuous process becomes the limit of our process as the time
subdivision becomes infinitesimal. 
We taxonomize this class of communication channels based on whether the
delays and noise are stochastic or adversarial; and based on how much
information each aspect has about the other when introducing its errors.
We analyze the limits of such
channels and reach somewhat surprising conclusions: The capacity
of a physical channel is finitely bounded
only if at least one of the two sources of error
(signal noise or delay noise) is adversarial. In particular the 
capacity is finitely bounded only if the delay is adversarial,
or the noise is adversarial and acts with knowledge of the stochastic
delay. 
If both error sources are stochastic,
or if the noise is adversarial and independent
of the stochastic delay, then the capacity of the associated physical
channel is infinite!
\end{abstract}

\newcommand{\keywords}[1]{{\noindent \bf Keywords: } #1}

\section{Introduction}

It seems to be a folklore assumption that any physical medium of
communication is constrained to communicating a finite number of
bits per unit of time. This assumption forms the foundations of
both the theory of communication \cite{Shannon} as well as the
theory of computing \cite{Turing}. The assumption also seems well-founded
give the theory of signal processing. In particular the work
of Shannon~\cite{Shannon:SP} explains reasons why such a statement may be
true.

Any physical channel (a copper wire, an optical fiber, vaccuum etc.)
in principle can be used by a {\em sender} to transmit a signal,
i.e., a function $f: [0,T] \to [0,1]$ for some time duration $T$. The
{\em receiver} receives some function $\tilde{f}: [0,T] \to [0,1]$,
which tends to be a noisy, distorted version of the signal $f$.
The goal of a communication system is to design encoders
and decoders that communicate reliably over this channel.
Specifically, one would like to find the largest integer $k_T$
such that there exist functions $E:\{0,1\}^{k_T} \to \{f:[0,T] \to [0,1]\}$
and $D:\{\tilde{f}:[0,T] \to [0,1]\} \to \{0,1\}^{k_T}$ such
that $\Pr_{\tilde{f}|E(m)} [D(\tilde{f}) \ne m] \to 0$ where
$m$ is chosen uniformly from $\{0,1\}^{k_T}$ and $\tilde{f}$ is
chosen by the channel given the input signal $E(m)$.
The {\em capacity of the channel}, normalized per unit of time,
is the $\lim\sup_{T \to \infty} k_T/T$.

In a typical such channel there are two possible source of ``infinity''.
The signal value $f(t)$, for any $t \in [0,T]$ is uncountably large
and if the channel were not ``noisy'' this would lead to infinity capacity,
even if time were discrete. But Shannon, in his works
\cite{Shannon,Shannon:SP}, points out that usually $f(t)$ is not
transmitted as is. Typical channels tend to add noise, typically
a random function $\eta(t)$, which is modeled as a normally distributed
random variable with mean zero and variance $\sigma^2$, and independent
across different instances of time $t$. He points out that after
this noise's effect is taken into account, the channel capacity is
reduced to a finite number (proportional to $1/\sigma^2$) per instant
of time.

Still this leaves a second possible way the channel capacity could be
infinite, namely due to the availability of infinitely many time slots.
This aspect has been considered before in the signal processing literature,
and the works of Nyquist~\cite{Nyquist'24}
and Hartley~\cite{Hartley'28} (see the summary in \cite{Shannon:SP})
once again point out that there is a finite limit. However the reason for
this finite limit seems more axiomatic than physical. Specifically, these
results come from the assumption that the signal $f$ is a linear
combination of a finite number of basis functions, where the basis functions
are sinusoids with frequency that is an integral multiple of some minimal
frequency, and upper bounded by some maximum frequency. This restriction
is then translated into a ``discretization'' result showing it suffices
to sample the signal at certain discrete time intervals, reducing the
problem thus to a finite one.

In this work we attempt to explore the effects of ``continuous time''
more in the spirit of the obstacle raised in the context of the
signal strength, namely that there is an obstacle also to assuming
that time is preserved strictly accross the communication channel.
We do so by introducing and studying a ``delay channel'' where
signals transmitted by the sender arrive somewhat asynchronously at
the receiver's end. We model and study this process
as the limit of a discrete process.

In our discrete model the sender/receiver get to discretize
time as finely as they wish, but there is uncertainty/unreliability
associated with the
delay between when a signal is sent and when it is received. Thus in this
sense, there is timing noise, that is similar in spirit to the signal
noise.
A signal that is sent at time $t$ is received at time $t + \eta(t)$
where $\eta(t)$ could be a random, or adversarial, amount of delay,
but whose typical amount is a fixed constant (independent of the
granularity of the discretization of time chosen by sender/receiver).
Note that this could permute the bits in the sequence sent by the
sender (or do more complex changes).
We consider the effect of this delay on the channel capacity.
For the sake of simplicity (and since this is anyway without loss
of generality) we assume sender only sends a sequence of 0s and 1s.
In addition to delays we also allow the channel to inject the usual
noise.

We discuss our model and results more carefully in Section~\ref{sec:prelims},
but let us give a preview of the results here. It turns out that
the question of when is the channel capacity finite is a function
of several aspects of the model. Note there are two sources of
error - the {\em signal error}, which we simply refer to as {\em noise},
and the {\em timing error}, which we refer to as {\em delay}. As either of these
error sources could be probabilistic or adversarial, we get four
possible channel models. Complicating things further is the
dependence between the two -- does either of the sources of error
know about the error introduced by the other? Each setting ends
up requiring a separate analysis. We taxonomize the many classes
of channels that arise this way, and characterize the capacity
of all the channels. The final conclusion
is the following: If the delays are adversarial, or if the 
delay is stochastic and the noise is adversarial and acts with knowledge 
of the delay, then the channel capacity is finite (Theorem~\ref{thm:finite}), else it is
infinite (Theorem~\ref{thm:infinite}). 
In particular if both sources are adversarial then the channel capacity
is finite; and perhaps most surprisingly and possibly the most
realistic setting, if both sources are probabilistic, then the channel
capacity is infinite: finer discretization always leads to increased
capacity.

\smallskip
\noindent
{\bf Organization:} Section~\ref{sec:prelims} formally describes our model and results.
In Sections~\ref{sec:finite} and \ref{sec:infinite}, we prove our results for finite and infinite channel
capacity regimes respectively. Finally, we give some concluding thoughts in
Section~\ref{sec:conclude}.

\section{Preliminaries, Model, and Results}
\label{sec:prelims}

\subsection{Continuous channels}

We start by describing the basic entities in a communication
system and how performance is measured. Most of the definitions
are ``standard''; the only novelty here is that we allow
sender/receiver to choose the ``granularity'' of time.
We first start with the standard definitions.

\paragraph{Channel(Generic)}
Given a fixed period of time $T$, a {\em signal} is a function
$f:[0,T] \to \R$. We say the signal is bounded if its range is $[0,1]$.
A time $T$ (bounded-input) channel is given by a
(possibly non-deterministic, possibly probabilistic, or
a combination) function $\channel_T : f \mapsto \tilde{f}$
whose inputs is a bounded signal $f:[0,T] \to [0,1]$
and output is a signal $\tilde{f}:[0,T] \to \R$.

A {\em probabilistic channel} is formally given by a transition probability
distribution which gives the probability of outputting
$\tilde{f}$ given input $f$. An {\em adversarial channel} is given by a
set of possible functions $\tilde{f}$ for each input $f$.
We use $\tilde{f} = \channel_T(f)$ as shorthand for $\tilde{f}$
drawn randomly from the distribution specified by $\channel_T(f)$
in the case of probabilistic channels.
For adversarial channels, we use the same notation
$\tilde{f} = \channel_T(f)$ as shorthand for
$\tilde{f}$ chosen adversarially (so as to minimize successful
communication) from $\channel_T(f)$.

Channels can be composed naturally, leading to interesting mixes
of adversarial and stochastic channels, which will lead to interesting
scenarios in this work.

\paragraph{Encoder/Decoder}
Given $T$ and message space $\{0,1\}^{k_T}$, a
time $T$ encoder is a function $E : m \mapsto
f$ where $m \in \{0,1\}^{k_T}$ and $f:[0,T] \to [0,1]$.
Given $T$ and message space $\{0,1\}^{k_T}$, a
time $T$ decoder is a function $D:f \mapsto m$
where $f:[0,T] \to \R$ and $m \in \{0,1\}^{k_T}$.
(More generally, encoders, channels, and decoders
should form composable functions.)

\paragraph{Success Criteria, Rate and Capacity}
The decoding error probability of the system $(E_T,D_T,\channel_T)$
is the quantity
$$\decerrT = \Pr_{m \from \{0,1\}^{k_T}, \channel} [ m \ne
D_T(\channel_T(E_T(m))) ].$$
We say that the communication system is {\em reliable} if
$\lim_{T \to \infty} \{\decerrT\} = 0$.

The (asymptotic) {\em rate} of a communication system is the limit
$\lim \sup_{T \to \infty} \{k_T/T\}$.
The {\em capacity of a channel}, denoted by $\mcap$,
is defined to be the supremum of the rate of the communication
system over all encoding/decoding schemes.

\subsection{Channel models}

We now move to definitions specific to our paper.
We study continuous channels as a limit of discrete channels.
To make the study simple, we restrict our attention to
channels whose signal strength is already discretized, and
indeed we will even restrict to the case where the channel
only transmits bits. The channel will be allowed to err,
possibly probabilistically or adversarially, and $\eps$
will denote the error parameter.

We now move to the more interesting aspect, namely the
treatment of time.
Our model allows the sender and receiver to divide every unit of
time into tiny
subintervals, which we call {\em micro-intervals},
of length $\mu = 1/M$ (for some large integer
$M$), and send arbitrary sequences of $M$ bits per unit of time.
This granularity is compensated for by the fact that
the channel is allowed to introduce relatively large,
random/adversarial, delays.
However the channel is allowed to introduce uncertain delays
into the system, where the delays average to some fixed
constant $\Delta$ which is independent of $\mu$. Given that
all aspects are scalable, we scale time so that $\Delta = 1$.
Again we distinguish between the adversarial case and
the probabilistic case. In the adversarial case
every transmitted symbol may be delayed by up to $1$ unit
of time (or by up to $M$ microintervals).  In the probabilistic
case
every transmitted symbol may be delayed by an amount
which is a random
variable distributed exponentially with mean $1$.
Finally, if multiple symbols end up arriving at the receiver
at the same instant of time, we assume the receiver receives
the sum of the value of the arriving symbols.

We describe the above formally:
\begin{description}
\item[Encoding:] For every $T$, the sender encodes $k_T$
bits as $MT$ bits by applying an encoding function
$E_T : \{0,1\}^{k_T} \to \{0,1\}^{MT}$.
The encoded sequence is denoted $X_1,\ldots,X_{MT}$.
\item[Noise:] The noise is given by a function
$\noise:[MT] \to \{0,1\}$. The effect of the noise is
denoted by the sequence $Z_1,\ldots,Z_{MT}$, where
$Z_j = X_j \oplus \noise(j)$. (We stress that $Z_j$'s 
are not necessarily ``seen'' by any physical entity ---
we just mention them since the notation is useful. Also,
the $\oplus$ is merely a convenient notation and
is not meant to suggest that the bits are elements of
some finite field. We will be thinking of the bits as
integers.)
\item[Delay:] The delay is modeled by a delay function
$\delay:[MT] \to \Z^{\geq 0}$ where
$\Z^{\geq 0}$ denotes the non-negative integers.
\item[Received Sequence]
The final sequence received by the receiver, on noise $\noise$
and delay $\delay$, is the sequence 
$Y_1,\ldots,Y_{MT}$, where $Y_i = \sum_{j\leq i {\rm ~s.t.~} j + \delay(j) =
i} Z_j$ and $Z_j = X_j \oplus \noise(j)$.
\item[Decoding]
The decoder is thus a function $D_T:(\Z^{\geq 0})^{MT} \to \{0,1\}^{k_T}$.
\end{description}

Note that while the notation suggests that the
noise operates on the input first, and then the delay acts on it,
we do not view this as an operational suggestion. Indeed the order
in which these functions ($\noise$ and $\delay$) are chosen will
be crucial to our results.

Our channels are thus described as a composition of
two channels, the noise-channel with parameter $\eps$,
denoted $N(\eps)$
and the delay-channel $D$. Since each of these can be probabilistic
or adversarial, this gives us four options. Furthermore
a subtle issue emerges which is: Which channel goes first?
Specifically if exactly one of the channels is adversarial,
then does it get to choose its noise/delay before or after knowing the
randomness of the other channel. We allow both possibilities
which leads syntactically to eight possible channels (though
only six of these are distinct).

{\bf Notation:} We use $D$ to denote the delay channel and
$N(\eps)$ to denote the noise channel with parameter $\eps$.
We use superscripts of $A$ or
$P$ to denote adversarial or probabilistic errors respectively.
We use the notation $X|Y$ to denote the channel $X$ goes first
and then $Y$ acts (with knowledge of the effects of $X$).
Thus the eight possible channels we consider are
$N^P | D^P$,
$D^P | N^P$,
$D^A | N^P$,
$N^A | D^P$,
$D^P | N^A$,
$N^P | D^A$,
$N^A | D^A$,
and $D^A | N^A$.

\subsection{Our results}

Given that the adversarial channels are more powerful than the corresponding
random channels, and an adversary acting with more information
is more powerful than one acting with less, some obvious
bounds on the capacity of these channels follow:

\begin{equation}
\label{eqn:chain1}
\mcap(D^A | N^A)
=\mcap(N^A | D^A)
\leq \mcap(N^P | D^A)
\leq \mcap(D^A | N^P)
\leq \mcap(D^P | N^P),
\end{equation}
and
\begin{equation}
\label{eqn:chain2}
\mcap(D^A | N^A)
\leq \mcap(D^P | N^A)
\leq \mcap(N^A | D^P)
\leq \mcap(N^P | D^P)
= \mcap(D^P | N^P).
\end{equation}
The equalities above occur because if both channels are adversarial,
or both are probabilistic, then ordering is unimportant.

Our main results are summarized by the
following two theorems.

\begin{theorem}[Finite Capacity Case]
\label{thm:finite}
For every positive $\eps$, the capacity of
the channels $D^A|N(\eps)^A$,\allowbreak $D^P|N(\eps)^A$,\allowbreak
$D^A|N(\eps)^P$,\allowbreak
and $N(\eps)^P|D^A$ are finite. That is, for every
one of these channel types, and $\eps > 0$ there
exists a capacity $C < \infty$ such that for every
$\mu$, a $\mu$-discretized encoder and decoder with
rate $R > C$, there exists a $\gamma > 0$ such that
the probability of decoding error $\decerr \geq \gamma$.
\end{theorem}

\begin{theorem}[Infinite Capacity Case]
\label{thm:infinite}
There exists a positive $\eps$ such that the
capacity of the channels $D^P|N(\eps)^P$
and $N(\eps)^A|D^P$ are infinite. That is, for every
one of these channel types, there exists an $\eps > 0$
such that for every finite $R$, there exists a $\mu$
and a $\mu$-discretized encoder and decoder achieving rate
$R$, with decoding error probability $\decerr \to 0$.
\end{theorem}

The theorems above completely characterize the case where
the capacity is infinite. 
The theorems show that the capacity
is infinite if either both channels are probabilistic (the most
benign case) or if the noise is adversarial but acts without
knowledge of the randomness of the probabilistic delay.
On the other hand, the channel
capacity is finite if the delay is adversarial, or if
the noise is adversarial and acts with knowledge of the
probabilistic delay.

Relying on the ``obvious'' inequalities given earlier, it suffices to
give two finiteness bounds and one ``infiniteness'' bound to
get the theorems above, and we do so in the next two sections.
Theorem~\ref{thm:finite} follows immediately from 
Lemmas~\ref{DPNA}~and~\ref{DANP} (when combined with
Equations~(\ref{eqn:chain1})~and~(\ref{eqn:chain2})).
Theorem~\ref{thm:infinite} follows immediately from 
Lemma~\ref{NADP} (again using 
Equations~(\ref{eqn:chain1})~and~(\ref{eqn:chain2})).

\section{Finite Capacity Regime}
\label{sec:finite}

In this section we prove that the capacity of our channels
are finite, when the delay channel is adversarial
(and acts without knowledge of the noise) or when the
noise is adversarial and acts with knowledge of the delay.
We consider the case of the adversarial noise first,
and then analyze the case of the random errors
adversarial delay first, and
then consider the case of the adversarial noise. 
In both
cases we use a simple scheme to show the capacity is limited.
We show that with high probability, the channel can force
the receiver to receive one of a limited number of signals.

The following simple lemma is then used to lower bound the probability
of error.

\begin{lemma}
\label{lem:capacity}
Consider a transmission scheme with the sender sending
message from a set $S$ with encoding scheme $E$, where the
channel $\channel$
can select a set $R$
of receiver signals such that
$$\Pr_{m \in S, \channel} [\channel(E(m)) \not \in R ] \leq \tau,$$
then the probability of decoding error is at least
$1 - (\tau + |R|/|S|)$.
\end{lemma}

\begin{proof}
Let $\tilde{R}$ denote the space of all received signals
and
fix the decoding function $D:\tilde{R} \to S$. Now consider the
event that a transmitted message $m$ is decoded correctly.
We claim this event occurs only if only of the two events
listed below occur:
\begin{enumerate}
\item $\channel(E(m)) \not\in R$, which happens with probability
at most $\tau$.
\item $m = D(r)$ for some $r \in R$, which happens with probability
at most $|R|/|S|$.
\end{enumerate}
If neither of the events listed above occur then the received
signal $r \in R$ and $D(r) \ne m$ implying the decoding is incorrect.
The lemma follows immediately.
\end{proof}

\subsection{Random Delay followed by Adversarial Noise $(D^P|N^A)$}

Here we consider the case where the delays are random, with expectation
$1$, and the noise is adversarial. In this section, it is useful
to view the delay channel as a queueing system, under the noise
channel's active control. To explain the queueing system, 
notice that exponential delays lead to a memoryless queue. At
each microinterval of time, a packet enters the queue (the new
bit sent by the sender). And then each packet in the queue chooses
to depart, independent of other packets, with probability $\mu$.
(Note that the exponential delay/memorylessness renders the packets
in the queue indistinguishable in terms of their arrival times.)

For the noise channel also, we will adopt a slightly different 
view. In principle, it is capable of looking at the entire
sequence of bits in the order in which they depart the queue, and
then decide which ones to flip. However our adversary will be much
milder. It will divide time into small intervals with the total number of intervals
being $N= O(T/\eps)$. In each interval it will ``hold'' most
arriving packets, releasing only those that are supposed to leave
the queue. If packets are released during the interval, the noise
adversary sets their value to $0$. With the remaining packets it
inserts them into the queue, with an integer multiple of $\eps M/c$ of
them being set to $1$ (and flipping a few bits to $0$ in the process as needed)
for some constant $c = c(\eps)$.
The remaining departures from the queue will then be transmitted 
untampered to the receiver. We will show that this departure 
process can be simulated by just the knowledge of the number
of $1$s injected into the queue at the end of each interval,
and the number of possibilities is just $(c+1)^N = c(\eps)^{O(T/\eps}))$
which is independent of $M$. The adversary will be able to carry out
its plan with probability $1 - \exp(-T)$, giving us the final result.
The following lemma and proof formalize this argument.

\begin{lemma}
\label{DPNA}
For every positive $\eps$, there exists a capacity
$C = C(\eps)$ such that the capacity of
the channel $D^P|N(\eps)^A$ is bounded by $C$.
Specifically, for every rate $R > C$,
for every $M$ (and $\mu = 1/M$),
every $T$, every $k_T > R\cdot T$,
and every pair of encoding/decoding functions
$E_T: \{0,1\}^{k_T} \to \{0,1\}^{MT}$
and
$D_T: (\Z^{\geq 0})^{MT} \to \{0,1\}^{k_T}$,
the decoding error probability
$\decerr = 1 - \exp(-T)$.
\end{lemma}

\medskip
\noindent
\begin{proof} We start with a formal description of the channel action, and then proceed 
to analyze the probability of decoding error and channel capacity.

\medskip
\noindent
\paragraph{Channel action:}
Let $X_1,\ldots,X_{MT}$ denote the $MT$ bit string being sent
be the sender. We will use $Z_j = X_j \oplus \noise(j)$ to 
denote the value of the
$j$th bit after noise (even though the noise acts after the
delay and so $Z_j$ may not be the $j$th bit received by the
receiver). We let $\delay:[MT] \to \Z^{\geq 0}$ denote the 
delay function.

Let $\eps' = \eps/5$ and $L = \eps' M$. 
The noise adversary partitions the $MT$ microintervals into
$T/\eps'$ intervals of length $L$ each, where the $i$th
interval $\Gamma_i = \{(i-1)L+1,\ldots,iL\}$. For every index $i \in \{1,\ldots,T/\eps'\}$, the
adversary acts as follows to set the noise function for
packets from $\Gamma_i$:
\begin{enumerate}
\item Let $n_i$ denote the Hamming
weight of the string $X_{(i-1)L+1}\ldots X_{iL}$, i.e., the weight
of the arrivals in the queue in interval $i$.
\item 
Let $\tilde{n}_i$ denote the rounding down of $n_i$
to an integer multiple of $\eps' \cdot L = (\eps')^2 \cdot M$ (we
assume all these are integers).
\item
Let $R_i$ denote the set of packets that arrive and leave
in the $i$th interval, i.e., $R_i = \{j \in \Gamma_i | j+\delay(j) \in
\Gamma_i\}$. 
\item
For every $j \in R_i$, the adversary sets $Z_j = 0$ (or $\noise(j) = X_j$).
Let $y_i = |R_i|$ and let $\hat{n}_i = \min\{\tilde{n}_i, L - y_i\}$.
\item
The adversary flips the minimum number of packets from $\Gamma_i \setminus R_i$
so that exactly $\hat{n}_i$ of these are ones.
\end{enumerate}
If at any stage the adversary exceeds its quota of $\eps MT$
errors it stops flipping any further bits.

\medskip
\noindent
{\bf Error Analysis:}
We claim first that the probability that the adversary stops
due to injecting too many errors is exponentially low.
This is straightforward to bound. Notice that the number
of bits of flipped in the $i$th interval due to early departures, 
is at most $y_i$, and $E[y_i] \leq
\frac12 \eps' L$. The number of bits flipped for packets 
that wait in the queue (i.e., from $\Gamma_i \setminus R_i$) is at most $\max\{y_i,n_i -
\tilde{n_i}\}$. Again the expectation of this is bounded by
the expectation of $y_i + (n_i - \tilde{n_i})$ which is at most
$\frac32(\eps') L$. Thus adding up the two kinds of errors, we find the 
expected number of bits flipped in the $i$th interval is at
most $\frac{5}2 \eps' L$. Summing over all intervals and applying
Chernoff bounds, we find the probability that we flip more than
$(5 \eps' = \eps)$-fraction of the bits is exponentially small in $T$.

\medskip
\noindent
{\bf Capacity Analysis:}
For the capacity analysis, we first note that the departure
process from the delay queue (after the $\noise$ function has
been set) is completely independent of the encoding $X_1,\ldots,X_{MT}$,
conditioned on $\tilde{n}_1,\ldots,\tilde{n}_{T/\eps'}$ and
on the event that the adversary does not exceed its noise bounds.
Indeed for any fixing of the $\delay$ function where the adversary
does not exceed the noise bound, the output of $D^P|N^A$ channel
on $X_1,\ldots,X_{MT}$ is the same as on the string
$\tilde{X}_1\ldots \tilde{X}_{MT}$, where for each $i$, the
string $\tilde{X}_{(i-1)L+1}\ldots \tilde{X}_{iL}$ is set to 
$1^{\tilde{n_i}} 0^{L-\tilde{n_i}}$.
Furthermore, note that the number of possible values of $\tilde{n_i}$
is at most $1/\eps'$.
We thus conclude that with all but exponentially small probability,
the number of distinct distributions received by the receiver (which
overcounts the amount of information received by the receiver) is
at most $(1/\eps')^{T/\eps'} = (1/\eps)^{O(T/\eps)}$.
An application of Lemma~\ref{lem:capacity} now completes the proof.

\end{proof}

\subsection{Adversarial Delay followed by Random Noise ($D^A|N^P$)}

\begin{lemma}
\label{DANP}
For every positive $\eps \le \frac 12$, there exists a capacity
$C = C(\eps)$ such that the capacity of
the channel $D^A|N(\eps)^P$ is bounded by $C$.
Specifically, for every rate $R > C$, there exists a $\gamma > 0$
and $T_0 < \infty$
such that
for every $M$ (and $\mu = 1/M$),
every $T \geq T_0$,
and every pair of encoding/decoding functions
$E_T: \{0,1\}^{k_T} \to \{0,1\}^{MT}$
and
$D_T: (\Z^{\geq 0})^{MT} \to \{0,1\}^{k_T}$,
the decoding error probability
$\decerr > \gamma$ if $k_T > R\cdot T$.
\end{lemma}

\noindent
{\bf Proof Idea:}
We give the capacity upper bound in two steps. In the first
step we create an adversarial delay function that attempts
to get rid of most of the ``detailed'' information being
sent over the channel. The effect of this delay function
is that most of the information being carried by the channel
in $M$ microintervals can be reduced to one of a constant
(depending on $\eps$) number of possibilities -- assuming
the errors act as they are expected to do. The resulting process
reduces the information carrying capacity of the channel
to that of a classical-style (discrete, memoryless) channel, 
and we analyze the capacity
of this channel in the second step. We give a few more details
below to motivate the definition of this classical channel.

We think of the delay function as a ``queue'', and the bits
being communicated as ``packets'' arriving/departing from this queue.
We call a packet a {\em $0$-packet} if it was a zero under the encoding
and as a {\em $1$-packet} if it was a one under the encoding. Note
that both types of packets, on release, get flipped with probability $\eps$
and the receiver receives one integer per time step representing
the total number of ones received.
The delay adversary clusters time into many large intervals and holds on
to all packets received during an interval, and releases most
of them at the end of the interval.
In particular if it releases $\tilde{n_0}$
$0$-packets  and $\tilde{n_1}$ $1$-packets at the end of an
interval, it makes sure that $\eps \tilde{n_0} + (1-\eps) \tilde{n_1}$
takes on one of a ``constant'' number of values independent of $M$.
(The actual value will be within $\pm \frac12$ of 
an integer multiple of $M/c$ due to integrality issues, but
in this discussion we pretend we get an exact multiple of $M/c$.)
Note that the quantity
$\eps \tilde{n_0} + (1-\eps) \tilde{n_1}$ denotes the expected
value of the signal received by the receiver when
$\tilde{n_0}$
$0$-packets  and $\tilde{n_1}$ $1$-packets are released, and so
we refer to this quantity as the {\em signature} of the interval. If the
errors were ``deterministic'' and flipped exactly the expected
number of bits, then the channel would convey no information beyond
the signature, and the total number of possible signatures over
the course of all intervals would dictate the number of possible
messages that could be distinguished from each other.

However the errors are not ``deterministic'' (indeed --- it is not
even clear what that would mean!). They are simply Bernoulli flips
of the bits being transmitted, and it turns out that different
pairs $(\tilde{n_0},\tilde{n_1})$ with the same signature can be
distinguished by the receiver due to the fact that they have
different variance. This forces us to quantify the information
carrying capacity of this ``signal-via-noise'' channel.

In the sequel, we first introduce this ``signal-via-noise'' channel
(Definition~\ref{def:svn-channel})
and bound its capacity (Lemmas~\ref{lem:danp-inf-capacity}~and~\ref{lem:danp-func-capacity}).
We then use this bound to give a proof of Lemma~\ref{DANP}.

\subsubsection{The ``signal-via-noise'' channel}

We introduce the ``signal-via-noise'' channel which is a discrete memoryless channel, 
whose novelty is in the fact that it
attempts to convey information using the variance
of the signal. We recall below some basic definitions from information
theory which we will use to bound the capacity of this channel. 
(These can also be found in \cite[Chapter 2]{CoverThomas}.)

Let $X$ be a random variable taking values from
some set $\calx$, and let $p_x$ denote the probability that $X = x$.
Then the {\em entropy} of $X$, denoted $H(X)$,
is the quantity $H(X) = \sum_{x \in \calx} p_x \log (1/p_x)$.

Let $X$ and $Y$ be jointly distributed random variables
with $X$ taking values from $\calx$ and $Y$ from $\caly$.
Let $p_{x,y}$ denote the probability that $X = x$ and
$Y = y$. For $y \in \caly$, let $H(X|y)$ denote the entropy of
$X$ conditioned on $Y = y$, that is, $H(X|y) = \sum_{x \in \calx} p_{x|y} \log (1/p_{x|y})$
where $p_{x|y} = p_{x,y}/(\sum_{z \in \calx} p_{z,y})$
denotes the probability that $X = x$ conditioned on $Y = y$.
Then the {\em conditional entropy} of $X$ given $Y$, denoted
$H(X|Y)$, is the quantity $H(X|Y) = \E_{y \in \caly} [ H(X|y) ]$.
The {\em mutual information} between $X$ and $Y$, denoted
$I(X;Y)$, is the quantity $I(X;Y) = H(X) - H(X|Y)$.
We will rely on the following basic fact.

\begin{proposition} {\rm \cite[Chapter 2]{CoverThomas}}
\begin{enumerate}
\item $H(X,Y) = H(X) + H(Y|X) = H(Y) + H(X|Y)$.
\item $I(X;Y) = I(Y;X) = H(Y) - H(Y|X)$.
\end{enumerate}
\end{proposition}

A discrete channel $\calc$ is given by a triple $(\calx,\caly,\calp)$, where
$\calx$ denotes the finite set of input symbols,
$\caly$ denotes the finite set of output symbols, and $\calp$ is a stochastic matrix with $\calp_{ij}$
denoting the probability that the channel outputs $j \in \caly$
given $i \in \calx$ as input. We use $\calc(X)$ to denote the output of this channel on input $X$.
The {\em information capacity} of such a channel is
defined to be the maximum, over all distributions $\cald$
on $\calx$, of the mutual information $I(X;Y)$
between $X$ drawn according to $\cald$ and $Y = \calc(X)$.

The information capacity turns out to capture the operational
capacity (or just capacity as introduced in Section~\ref{sec:prelims})
of a channel when it is used many times (see
Lemma~\ref{lem:danp-func-capacity} below).
Our first lemma analyzes the information capacity of the
``signal-via-noise'' channel, which we define formally below. 

In what follows, we fix positive integers $M$ and
$c$ and a rational $\eps$.

\begin{definition}
\label{def:svn-channel}
For integers $M, c$ and $\eps > 0$, the collection
of $(M,\eps,c)$-channels is given by $\{\calc_\mu | \mu \in
\{M/c,\ldots,M\} \}$, where the channel
$\calc_\mu = (\calx_\mu, \caly_\mu, \calp_\mu)$ is
given by
$$\calx_\mu = \{(a,b) \in \Z^{\geq 0} \times \Z^{\geq 0} \mid
\mu-\frac12 < \eps a + (1-\eps) b \leq \mu+\frac12, 
0 \leq a+b \leq M\}, \caly_\mu = \{0,\ldots,M\},$$
and $\calp_\mu$ is the distribution that, on input $(a,b)$,
outputs the random variable $Y = \sum_{i=1}^a U_i +
\sum_{j=1}^b V_j$, where the $U_i$'s and $V_j$'s are
independent Bernoulli random variable with $\E[U_i] = \eps$
and $\E[V_j] = 1 - \eps$.
\end{definition}

Note that the expectation of the output of the channel $\calc_\mu$
is roughly $\mu$, and the only ``information carrying capacity'' is derived
from the fact that the distribution over $\{0,\ldots,M\}$ is
different (and in particular has different variance) depending
on the choice of $(a,b) \in \calx_\mu$.
The following lemma shows that this information carrying
capacity is nevertheless bounded as a function of $\eps$ and $c$
(independent of $M$). Later we follow this lemma with a standard
one from information theory showing that the information capacity
does bound the functional capacity of this channel.

\begin{lemma}
\label{lem:danp-inf-capacity}
For every $0 < \eps \le \frac 12$ and $c < \infty$, there exists
$C_0 = C_0(c,\eps)$
such that for all $M$ the information capacity of 
every $(M,\eps,c)$-channel is at most $C_0$.
\end{lemma}
\begin{proof}
The lemma follows from the basic inequality for any pair of random
variables $X$ and $Y$ that $I(X;Y) = H(Y) - H(Y | X)$,
where $H(\cdot)$ denotes the entropy function and $H(\cdot | \cdot)$
denotes the conditional entropy function.
Thus to upper bound the capacity it suffices to give a lower
bound on $H(Y|X)$ and an upper bound on $H(Y)$.

We prove below some rough bounds that suffice for us.
Claim~\ref{clm:channel-one} proves $H(Y|X) \geq \frac12\log_2 M
- C_1(c,\eps)$ and
Claim~\ref{clm:channel-two} proves $H(Y) \leq \frac12\log_2 M
+ C_2(c,\eps)$.
It immediately follows that the capacity of the channel $\calc_\mu$
is at most $C_1(c,\eps) + C_2(c,\eps)$.
We now proceed to prove Claims~\ref{clm:channel-one}
~and~\ref{clm:channel-two}.

\begin{claim}
\label{clm:channel-one}
There exists $C_1(c,\eps)$ such that
for every $(a.b) \in \calc_\mu$,
$H(Y|X = (a,b)) \geq \frac12\log_2 M - C_1(c,\eps)$.
\end{claim}
\begin{proof}
This part follows immediately from the following claim
which asserts that for every $j \in \caly_\mu$,
$\Pr[Y = j|X = (a,b)] \leq 8(c/\eps)^{3/2} M^{-\frac12}$.
We thus conclude

$$H(Y|X = (a,b)] \geq \min_j \log \left( \frac{1}{Pr[Y = j|X = (a,b)]} \right)
\geq \frac12 \log M - \frac 32 \log \left(\frac{c}{\eps}\right) - 3.$$

\begin{claim}
For every $j \in \caly_\mu$,
$\Pr[Y = j|X = (a,b)] \leq 8(c/\eps)^{3/2} M^{-\frac12}$.
\end{claim}
\begin{proof}
We use the Berry-Ess\'een theorem~\cite[Chapter 16]{FellerV2'67},
and in particular the following version of the theorem.

If $Z_1,\ldots,Z_\ell$ are random variables with mean
zero such that $\sum_{i=1}^\ell E[Z_i^2] \leq \sigma^2$
and $\sum_{i=1}^\ell E[|Z_i|^3] \leq \rho$, then for
every $\alpha \in \R$,
$$\left|\Pr\left[\frac1\sigma\cdot \left(\sum_{i=1}^\ell Z_i \right) \leq \alpha\right] - \Phi(\alpha)\right|
\leq \rho/\sigma^3.$$
An immediate implication is that for $\alpha \leq \beta$, we have
$$ \Pr\left[ \alpha \leq \frac1\sigma \cdot \left(\sum_{i=1}^\ell
Z_i \right) \leq \beta \right] \leq \Phi(\beta) - \Phi(\alpha) +
2\rho/\sigma^3.$$
In our setting, $\ell = a+b$, and $Z_i = U_i - \eps$
for $1 \leq i \leq a$ and $Z_i = V_i - (1 -\eps)$ for
$a+1 \leq i \leq \ell$. We have $\sigma^2 = (a+b)\eps(1-\eps)
\geq \eps(1-\eps)M/c\geq \eps M/(2c)$. Finally $\rho$ can be crudely upper bounded
by $M$ (since $\ell \leq M$ and $-1 \leq Z_i \leq 1$).
Let $\tilde{\mu} = \eps a + (1-\eps) b$.
Then, if we set $\alpha = \beta = (j-\tilde{\mu})/\sigma$,
we find that $\Pr[\sum_{i=1}^\ell Z_i = j] \leq  8(c/\eps)^{3/2} M^{-\frac12}$.
\end{proof}

\end{proof}

\begin{claim}
\label{clm:channel-two}
There exists $C_2$ such that
for every random variable $X$ supported on $\calx_\mu$
and $Y = \calc_\mu(X)$,  we have
$H(Y) \leq \frac12\log_2 M + C_2$.
\end{claim}

\begin{proof}
Let $\sigma$ be such that
$\sigma^2 = M\eps(1-\eps)$
is an upper bound on the variance of $Y$ conditioned on $X$.
We use this upper bound to bound the probability of $Y$
being too far from $\mu$, and in turn use this to bound its
entropy.

We partition $\caly_\mu$ into a sequence of sets
$S_0,S_{2},S_{\sigma-1},\sinf$ as defined below.
$$S_0 = \{j \in \caly_\mu \mbox{ s.t. } |j - \mu | < 2\sigma\},$$
$$\sinf = \{j \in \caly_\mu \mbox{ s.t. } |j - \mu | \geq \sigma^2\},$$
$$\mbox{ and }
S_i = \{j \in \caly_\mu \mbox{ s.t. }
i \cdot \sigma \leq |j - \mu | < (i+1)\cdot \sigma\},$$
for $i \in \{2,\ldots,\sigma-1\}$.

Let $h_0,h_2,\ldots,h_{\sigma-1}$ and $\hinf$ denote the contribution
of $S_0,S_{2},\ldots,S_{\sigma-1}$ and $\sinf$ to the entropy of $Y$,
i.e., $h_i = \sum_{j \in S_i} \Pr[Y = j] \log \frac1{\Pr[Y=j]}$
(for $i \in \{0,2,\ldots,\sigma-1,\infty\}$).
Similarly, let $p_i = \Pr[Y \in S_i]$.
Note that we have $H(Y) = h_0 + \sum_{i = 2}^{\sigma-1} h_i + h_\infty$
and we bound these separately below, using rough approximations
on $p_i$.

We start with a basic fact. For any set $S$, let $\Pr[Y \in S]\leq p_S$ for some 
$p_S \in (0,1]$. Then
$$
h_S = \sum_{j \in S} \Pr[Y = j] \log \frac1{\Pr[Y=j]}
\leq p_S \log (|S|/p_S).
$$
(Follows easily from the convexity of the
entropy function and Jensen's inequality.)

This immediately yields our first bound, using $|S_0| \leq 4\sigma$
and $p_0 \leq 1$. We have
\begin{equation}
\label{eqn:h0}
h_0 \leq p_0 \log (4\sigma/p_0) \leq p_0 \log \sigma + 2
\end{equation}

For the remaining parts we use the following Chernoff-like
bound from \cite[Theorem~7.2.1]{MatousekVondrak}\footnote{Among the many
such bounds available, this one allows variables to be non-identically
distributed.}
$$\Pr[|Y - \mu| \geq \ell \sigma ] \leq e^{-3\ell^2/4} \leq 2^{-\ell},$$
for $2 \leq \ell \leq \sigma$.
Thus $p_i \le 2^{-i}$ for $i \in \{ 2,..., \sigma-1\}$, and $\pinf \leq 2^{-\sigma}$.

For $i \in \{ 2,..., \sigma-1\}$, we can thus bound $h_i$ by

\begin{equation}
\label{eqn:hi}
h_i \leq p_i \log (\sigma/p_i) \leq p_i \log \sigma + i 2^{-i},
\end{equation}

and, bound $\hinf$ by

\begin{equation}
\label{eqn:hinf}
\hinf \leq \sigma \cdot 2^{-\sigma} = O(\exp(-\sqrt{M})) = O(1)
\end{equation}

Combining Equations~\ref{eqn:h0},~\ref{eqn:hi},~and~\ref{eqn:hinf}
we get
\begin{eqnarray*}
H(Y)
& = & h_0 + \sum_{i=2}^{\sigma-1} h_i + \hinf\\
& \leq & p_0 \log \sigma + 2 + \sum_{i=2}^{\sigma-1} (p_i \log
\sigma + i 2^{-i}) + O(1) \\
& \leq & (p_0 + \sum_{i=2}^{\sigma-1} p_i) \log \sigma
+ \sum_{i=0}^{\infty} i 2^{-i} + O(1) \\
& \leq & \log \sigma + O(1)
\end{eqnarray*}

The claim now follows from the fact that $\sigma \leq \sqrt{M}$.
\end{proof}

\end{proof}

Our analysis of the $D^A|N^P$ channel immediately yields a lower
bound on the ``operational capacity'' of any
sequence of channels $\{\calc_{\mu_i}\}_{i=1}^N$.
Standard bounds in information theory (see, for instance,
\cite[Chapter 8, Theorem 8.7.1]{CoverThomas}) imply immediately
that a bound on the capacity also implies that any attempt to
communicate at rate greater than capacity lead to error with
positive probability. We summarize the resulting consequence below.
(We note that while the theorem in \cite{CoverThomas} only considers
a single channel and not a collection of channels, the proof goes
through with only notational changes to cover a sequence of channels.)

\begin{lemma}
\label{lem:danp-func-capacity}
Transmission at
rate $R$ greater than $C_0$, the information capacity,
leads to error with positive probability.
More precisely, for any $0 < \eps \le \frac 12$, let $C_0 = C_0(\eps,c)$ be an upper bound on
the information capacity of a collection of channels
$\{\calc_{\mu} | \mu\}$.
Then for every $R > C_0$, there exists a $\gamma_0 > 0$
and $N_0 < \infty$ such that for every $N\geq N_0$
the following holds:
For every sequence $\{\calc_{\mu_i}\}_{i=1}^N$
of $(M,\eps,c)$ channels, and every encoding and decoding pairs
$E:\{0,1\}^{RN} \to \prod_{i=1}^N \calx_{\mu_i}$
and $D:\{0,\ldots,M\}^N \to \{0,1\}^{RN}$, the
probability of decoding error
$\decerr \geq \gamma_0$.
\end{lemma}

\subsubsection{\protect Proof of Lemma~\ref{DANP}}

\medskip
\noindent
\begin{proof} We now formally describe the delay adversary
and analyze the channel capacity.
Let $c = 4/\eps$ and let $C_0 = C_0(\eps,c)$ be the bound on
the capacity of $(M,\eps,c)$-channels $\calc_\mu$ from
Lemma~\ref{lem:danp-inf-capacity}. We prove the lemma
for $C(\eps,c) = 2(C_0 + \log c)$.

\smallskip
\noindent
{\bf Delay:}
Let $X_1,\ldots,X_{MT}$ denote the encoded signal the sender sends.
The noise channel picks $\noise(j)$ independently for each $j$
with $\noise(j)$ being $1$ w.p. $\eps$.
We now describe the action of the delay channel (which acts without
knowledge of $\noise$).

We divide time into $2T$ intervals, with the $i$th interval
denoted $\Gamma_i = \{(i-1)(M/2) + 1,\ldots, i(M/2)\}$.
Let $n_1(i) = \sum_{j \in \Gamma_i} X_j$ and
$n_0(i) = M/2 - n_1(i)$ denote the number of $1$-packets and
$0$-packets that arrive in the queue in the $i$th interval.
The delay adversary acts as follows:
\begin{enumerate}
\item Initialize $n'_1(1) = n_1(1)$ and $n'_0(1) = n_0(1)$.
\item For $i = 1$ to $2T$ do the following:
\begin{enumerate}
\item If $n_1(i) \geq n_0(i)$ then
set $\tilde{n}_0(i) = n'_0(i)$ and
round $n'_1(i)$ down to $\tilde{n}_1(i)$ so that $(1-\eps)\tilde{n}_1(i)
+ \eps \tilde{n}_0(i)$ is  within $\frac12$ of the nearest
integer multiple of $M/c$.
\item Else let
$\tilde{n}_1(i) = n'_1(i)$ and
round $n'_0(i)$ down to $\tilde{n}_0(i)$ so that $(1-\eps)\tilde{n}_1(i)
+ \eps \tilde{n}_0(i)$ is within $\frac12$ of the nearest integer multiple of $M/c$.
\item Finally set $n'_0(i+1) = n_0(i+1) + n_0(i) - \tilde{n}_0(i)$.
and $n'_1(i+1) = n_1(i+1) + n_1(i) - \tilde{n}_1(i)$.
\item
At the end of interval $i$, output $\tilde{n}_0(i)$ $0$-packets
and $\tilde{n}_1(i)$ $1$-packets from the queue to the noise
adversary. Formally, the delay channel outputs a set $\Lambda_i$
of packets that are to be released at the end of interval $\Gamma_i$,
where $\Lambda_i$ includes all packets that arrived in $\Gamma_{i-1}$
but were not included in $\Lambda_{i-1}$.
\item The noise adversary simply flips the bits according
to the noise function and outputs the sum of these bits.
Specifically it sets $Z_j = X_j + \noise(j)$
and outputs $Y_i = \sum_{j \in \Lambda_i} Z_j$.
\end{enumerate}
\end{enumerate}

\paragraph{Analysis:}
We start by establishing that the delay adversary 
never delays any packet by more than $M$ microintervals. Note that the number of 
packets that arrive in interval $i$, but are not released at the
end of the interval
is given by $(n'_1(i) - \tilde{n}_1(i)) + (n'_0(i) - \tilde{n}_0(i))$.
One of the two summands is zero by construction, and the other
{\bf is at most $M/(\eps c) \leq M/4$ by our construction}.
Since the total number of packets arriving in an interval
is $M/2$, this ensures that the total number released
in an interval is never more than $3M/4 \leq M$ (as required
for an $(M,\eps,\mu)$-channel). Next we note that packets
delayed beyond their release interval do get released
in the next interval. Again, suppose $n_0(i) > n_1(i)$.
Then all $1$-packets are released in interval $i$. And the
number of $0$-packets held back is at most $\frac{M}{\eps c} \leq
M/4$ which is less than $n_0(i)$ the total number of $0$-packets
arriving in interval $\Gamma_i$. Thus the adversary never delays
any packet more than $M$ microintervals, and the number of
packets released in all intervals
(except the final one) satisfy $\eps \tilde{n}_0(i) 
+ (1-\eps)\tilde{n}_1(i)$ in an integer multiple of $M/c$.

For an encoded message $X_1,\ldots,X_{MT}$, let
$\mu_i = [\epsilon \tilde{n}_0(i) + (1-\eps)\tilde{n}_1(i)]$,
where the notation $[x]$ indicates the nearest integer to $x$,
denote the {\em signature} of the $i$th interval;
and let $\vec{\mu} = (\mu_1,\ldots,\mu_{2T})$ denote its
signature.
Note that $\mu_i$ takes one of at most $c$ distinct values
(since it is between $M/c$ and $M$ and always an integer multiple of
$M/c$).
Thus the number of signatures is at most $c^{2T}$.

Now since the total number of distinct messages is $2^{k_T}$,
the average number of messages with a given signature sequence
is at least $2^{k_T}/c^{2T}$.
Furthermore, with probability at least $1-\delta$, a random
message is mapped to a signature sequence with at least
$\delta 2^{k_T}/c^{2T}$ preimages. Suppose that
such an event happens. Then, using the fact that 
$R > 2 (C_0 + \log c) - \frac{1}T \log \delta$, we argue below that 
conditioned on this event the probability of correct decoding
is at most $1 - \gamma_0$ (where $\gamma_0 > 0$ is the constant from
Lemma~\ref{lem:danp-func-capacity}). This yields the lemma
for $\gamma = (1 - \delta)\gamma_0$.

To see this, note that the signal $Y_1,\ldots,Y_{2T}$ 
received by the receiver is exactly the output
of the channel sequence $\{\calc_{\mu_i}\}_{i=1}^{2T}$
on input $\tilde{X}_1,\ldots,\tilde{X}_{2T}$ where $\tilde{X}_i = 
(1-\eps)\tilde{n}_1(i) + \eps \tilde{n}_0(i)$. 
If the receiver decodes the message (more precisely, its encoding)
$X_1,\ldots,X_{MT}$ correctly from $Y_1,\ldots,Y_{2T}$, then 
we can also compute
the sequence $\tilde{X}_1,\ldots,\tilde{X}_{2T}$ correctly
(since the delay adversary is just a deterministic function
of its input $X_1,\ldots,X_{MT}$). Thus correct decoding of 
the $D^A|N^P(\eps)$ channel
also leads to a correct decoding of the 
channel sequence $\{\calc_{\mu_i}\}$. 
But the number of distinct messages being transmitted
to this channel is $\delta 2^{k_T}/c^{2T}$. Denoting
this by $2^{\tilde{R}\cdot 2 \cdot T}$ and using the
fact that $\tilde{R} > C_0$, we get that the channel must
err with probability at least $\gamma$.
\end{proof}

\section{Infinite Capacity Regime}
\label{sec:infinite}

In this section we show that the capacity of the channel with
adversarial noise followed by random delay $(N^A|D^P)$ is infinite. Specifically, we establish the following result:

\begin{lemma}
\label{NADP}
There exists a positive $\eps$, such that the capacity
of the channel $N(\eps)^A|D^P$ is unbounded.
Specifically, for every rate $R$,
there exists a constant $M$ (and $\mu = 1/M$),
such that for sufficiently large $T$,
there exist encoding and decoding functions
$E_T: \{0,1\}^{k_T} \to \{0,1\}^{MT}$
and
$D_T: (\Z^{\geq 0})^{MT} \to \{0,1\}^{k_T}$,
the decoding error probability
$\decerr \leq \exp(-T)$, with $k_T = R\cdot T$.
\end{lemma}

\paragraph{Proof Idea:}
The main idea here is that the encoder encodes a $0$ by a series
of $0$s followed by a series of $1$s and a $1$ by a series of
$1$s followed by a series of $0$s. Call such a pair of series a
``block''. If the noisy adversary doesn't
corrupt too many symbols within such a block (and it can't afford
to do so for most blocks), then the receiver can distinguish
the two settings by seeing if the fraction of 1s being received
went up in the middle of the block and then went down, or the
other way around. This works with good enough probability (provided
the delay queue has not accumulated too many packets) to allow
a standard error-correcting code to now be used by sender and
receiver to enhance the reliability.

\medskip
\noindent
\begin{proof}
We will prove the lemma for $\eps < 1/64$ below.
\footnote{For clarity of exposition, we do not make any attempt to optimize the bound on the value of $\eps$.}
Let $k = k_T = RT$. We will set $M = O(R^5)$.
Let $L = M^{4/5}$, and $L' = M^{3/4}$,
$\Gamma_i = \{(i-1)L+1,\ldots,iL\}$,
and $\Gamma'_i = \{iL - L'+1,\ldots,iL\}$.
As a building block for
our sender-receiver protocol, we will use a pair
of classical encoding and decoding algorithms, $E'$ and $D'$, that
can handle up to $5/24$-fraction of adversarial
errors. (Note that $5/24$ could be replaced with any constant
less than $1/4$.)
In particular, for each message $m \in \{ 0,1 \}^k$, the algorithm $E'$ outputs
an encoding $E'(m)$ of length $N = \Theta(k)$
such that for any binary string
$s$ of length $N$ that differs from $E'(m)$ in at most $(\frac{5}{24})N$ locations,
$D'(s) = m$. We now describe our encoding and decoding protocols.

\medskip
\noindent
{\bf Sender Protocol:} The encoding $E = E_T$ works as follows.
Let $m \in \{ 0,1 \}^k$ be the message that the sender wishes
to transmit. The encoding $E(m)$ simply replaces every $0$
in $E'(m)$ with the string $0^L 1^L$,
and each $1$ in $E'(m)$ with the string $1^L 0^L $.
Thus $E(m)$ is a string of length $2 L N = MT$.
The sender transmits the string $E(m)$ over the channel.

\medskip
\noindent
{\bf Receiver Protocol:} Recall that the receiver receives, at
every microinterval of time $t \in [MT]$ the quantity
$Y_t = \sum_{j \leq t | j + \delay(j) = t} X_j \oplus \noise(j)$.
For an interval $I \subseteq [MT]$, let $Y(I) = \sum_{j \in I} Y_j$.
The decoding algorithm $D = D_T$, on input $Y_1,\ldots,Y_{MT}$
works as follows:
\begin{enumerate}
\item For $i = 1$ to $N$ do:
\begin{enumerate}
\item Let $\alpha_i = Y(\Gamma'_{2i-1})/L'$.
\item If $Y(\Gamma'_{2i}) - Y(\Gamma'_{2i-1}) \leq (-\alpha_i +
\frac12)\cdot M^{11/20}$ then set $w_i = 1$, else set $w_i = 0$.
\end{enumerate}
\item Output $D'(w)$.
\end{enumerate}

\paragraph{Analysis:}
By the error-correction properties of the pair $E',D'$, it
suffices to show that for $(19/24)$-fraction of
the indices $i \in [N]$, we have $w_i = E'(m)_i$.

Fix an $i \in [N]$ and let $Q_i$ denote the number of $1$'s in 
the queue at the beginning of interval $\Gamma'_{2i-1}$.
We enumerate a series of ``bad events'' for interval $i$
and show that if
none of them happen, then $w_i = E'(m)_i$. Later we show
that with probability ($1 - \exp(-T)$) the number of
bad $i$'s is less than $(5/24)N$, yielding the lemma.

We start with the bad events:
\begin{description}
\item[$\calE_1(i)$:] $Q_i > cM$ (for appropriately chosen constant $c$).
We refer $i$ as {\em heavy} (or more specifically $c$-heavy) if 
this happens.

\item[$\calE_2(i)$:] The number of errors introduced by the 
adversary in the interval $\Gamma_{2i}$ is more than 
$16\eps L$. We refer to $i$ as {\em corrupted} if this happens.

\item[$\calE_3(i)$:] $i$ is not $c$-heavy but one of $Y(\Gamma'_{2i-1})$ or
$Y(\Gamma'_{2i})$ deviates
from its expectation by more than $\omega(M^{1/2})$.
We refer to $i$ as {\em deviant} if this happens. 
\end{description}

In the absence of events $\calE_1$, $\calE_2$, $\calE_3$, we
first show that $w_i = E'(m)_i$.
Denote $i$ to be a $1$-block if $E'(m)_i = 1$ and a $0$-block otherwise.
To see this, we first compute the expected values of $Y(\Gamma'_{2i-1})$,
and $Y(\Gamma'_{2i})$ conditioned on $i$ being a $0$ block and $i$ being
a $1$ block. (We will show that these expectations differ by roughly
$M^{11/20}$, and this will overwhelm the deviations allowed for
non-deviant $i$'s.)

We start with the following simple claim.

\begin{claim}
\label{claim:release_prob}
Let $\ell_1, \ell_2$ be a pair of non-negative integers, and let ${\cal E}$
denote the event that a packet $p$ that is in the delay queue at some 
time $t$
leaves the queue during the interval $\{t + \ell_1+1,\ldots, t + \ell_1 + \ell_2\}$. Then

$$   \left( 1 - \frac{\ell_1}{M} \right) \left(\frac{\ell_2}{M} - \frac{\ell_2^2}{M^2} \right)   \le \Pr[{\cal E}]  \le \left( 1 - \frac{\ell_1}{M} + \frac{\ell_1^2}{M^2} \right) \left( \frac{\ell_2}{M} \right).$$
Thus if $\ell_1 = 0$ and $\ell_2 \ll M$, then
$ \frac{\ell_2}{M} - O\left(\left(\frac{\ell_2}{M} \right)^2\right) 
\le \Pr[{\cal E}]  \le \frac{\ell_2}{M}.$
\end{claim}
\begin{proof}
Note that
$$ \Pr[{\cal E}] = \left(1 - \frac{1}{M} \right)^{\ell_1} \left( 1 - \left(1 - \frac{1}{M} \right)^{\ell_2} \right).$$

Using the fact that for any non-negative integer $\ell$,

$$  1 - \frac{\ell}{M}  \le \left(1 - \frac{1}{M} \right)^{\ell} \le 
1 - \frac{\ell}{M} + \frac{\ell^2}{M^2},$$

we get the bounds in the claim.
\end{proof}

Let $Q_i = \tilde{\alpha}\cdot M$. We now analyze
the expectations of the relevant $Y(\cdot)$'s. We analyze
them under the conditions that $\tilde{\alpha}$ is bounded by
the constant $c$ (i.e. $i$ is not heavy) and that $i$ is not corrupt.
\begin{description}
\item[{\protect $\E[Y(\Gamma'_{2i-1})]$}:]
The probability that a single packet leaves the queue in
this interval is roughly $L'/M + O((L'/M)^2)$ (by Claim~\ref{claim:release_prob}
above). The expected number of packets that were in the
queue at the beginning of $\Gamma'_{2i-1}$ that leave the queue
in this interval is thus $(\tilde{\alpha}\cdot M \cdot L'/M) \pm
O((L')^2/M)$. Any potential new packets that arrive during
this phase contribute another $O((L')^2/M)$ potential packets,
thus yielding 
$\E[Y(\Gamma'_{2i-1})] = \tilde{\alpha} L' \pm O(\sqrt{M}) = \tilde{\alpha}M^{3/4} \pm O(\sqrt{M})$.
\item[{\protect $\E[Y(\Gamma'_{2i})]$} when $i$ is a $1$-block:]
Recall that a $1$-block involves transmission of $1$s in
$\Gamma_{2i-1}$ and $0$s in $\Gamma_{2i}$. With the 
adversary corrupting up to $16\eps L$ packets in $\Gamma_{2i}$
and the addition of $L'$ new $1$s in the interval $\Gamma'_{2i-1}$,
at most $L' + 16\eps L$ new ones may be added to the queue at the
beginning of the interval $\Gamma'_{2i}$. Using
Claim~\ref{claim:release_prob}
with $\ell_1 = L$ and $\ell_2 = L'$
to the $Q_i$ packets from the beginning of interval
$\Gamma'_{2i-1}$, and with $\ell_1=0$ and $\ell_2 = L'$
to the new packets that may have been added, we
get that

\begin{eqnarray*}
\E[Y(\Gamma'_{2i})] & \leq &
\left(1 - \frac{L}{M} + \frac{L^2}{M^2} \right)\cdot \left(\frac{L'}{M}\right)\cdot\tilde{\alpha}\cdot M + 
\left( \frac{L'}{M} \right)\cdot (16\eps L + L')\\
& = & \tilde{\alpha} M^{3/4} - (\tilde{\alpha} - 16\eps) M^{11/20}
+ O(\sqrt{M}).
\end{eqnarray*}
\item[{\protect $\E[Y(\Gamma'_{2i})]$} when $i$ is a $0$-block:]
In this case the $\Gamma_{2i-1}$ is all $0$s and
$\Gamma_{2i}$ is all $1$s. So the number of 
$1$s seen in the $\Gamma'_{2i}$ should be more than 
the number of $1$s seen in the
$0$-block case. In this case, the number of new $1$s added
to the queue in the intervals $\Gamma'_{2i-1}$ and $\Gamma_{2i} -
\Gamma'_{2i}$ is lower bounded
by $L - L' - 16\eps L$. Using Claim~\ref{claim:release_prob}
again to account for the departures from $Q_i$
as well as the new arrivals in the interval $\Gamma'_{2i}$, we get

\begin{eqnarray*}
\E[Y(\Gamma'_{2i})] & \geq & 
(1 - \frac{L}{M})\cdot \left(\frac{L'}{M} - O\left( \frac{L'^2}{M^2} \right) \right)\cdot 
\tilde{\alpha} \cdot M \\
& & + \left((1-16\eps)L-L' \right) \cdot \left( 1 - \frac{L}{M} \right) \cdot \left( \frac{L'}{M} - O\left(\frac{L'^2}{M^2} \right) \right) \\
&= & \tilde{\alpha} M^{3/4} - (\tilde{\alpha} - (1-16\eps))\cdot M^{11/20} +
O(\sqrt{M}).
\end{eqnarray*}
\end{description}

Putting the above together we see that 
$\E[Y(\Gamma'_{2i-1}) - Y(\Gamma'_{2i})]$ has a leading
term of $\tilde{\alpha} M^{3/4}$ in both cases ($i$ being
a $0$-block or $i$ being a $1$-block), but the second
order terms are different, and these are noticeably different.
Now, if we take into account the fact that the event $\calE_3(i)$
does not occur ($i$ is not deviant), then we conclude that the 
deviations do not alter even the second order terms.
We thereby conclude that if none of the events
$\calE_1(i)$ or $\calE_2(i)$ or $\calE_3(i)$ occur, then $w_i = E'(m)_i$.

We now reason about the probabilities of the three events.
The simplest to count is $\calE_2(i)$. By a simple averaging argument,
at most $(1/8)$th of all indices $i$ can have $i$ corrupt, since
the total number of noise errors is bounded by $\eps(2LN)$, and so the probability
of $\calE_2(i)$ is zero on at least $(7/8)$th fraction of indices.
$\calE_3(i)$ can be analyzed using standard tail inequalities.
Conditioned on $i$ being not $c$-heavy, 
each $Y(\cdot)$ is a sum of at most $(cM + L + L')$ independent random variables (each
indicating whether a given packet departs queue in the specified interval).
The probability that this sum deviates from its expectation by
$\omega(\sqrt{M})$ is $o(1)$. 
Thus, the probability that $\calE_3(i)$ happens for more than a $(1/24)$th fraction 
of indices $i$, can again be bounded by $\exp(-T)$ by Chernoff bounds.

The only remaining event is $\calE_1(i)$. Lemma~\ref{lem:heavy} below
shows that we can pick $c$ large enough to make
sure the number of heavy $i$'s is at most a $(1/24)$th fraction of all $i$s,
with probability at least $1 - \exp(-T)$.
We conclude that with probability at least $1 - \exp(-T)$
the decoder decodes the message $m$ correctly.
\end{proof}

\begin{lemma}
\label{lem:heavy}
For every $\delta > 0$, there exists a $c = c(\delta)$ such that
the probability that more than $\delta$-fraction of the 
indices $i$ are $c$-heavy is at most $e^{-(MT)/4}$.
\end{lemma}

\begin{proof}
Recall that an interval $i$ is {\em $c$-heavy} if $Q_i > cM$.
We will show that the lemma holds for $c = 4/\delta$.

For each packet $j$, recall that $\delay(j)$ indicates the number of microintervals 
for which the packet $j$ stays in the queue. 
Let $W = \sum_{j} \delay(j)$. 
We will bound the probability that $W$ is ``too large'' and then use this to 
conclude that the probability that too many intervals are heavy is small.

Note $W$ is the sum of $MT$ 
identical and independent geometric variables (namely the $\delay(j)$'s) with expectation of each being $M$. 
Thus the probability that $W > K$ (for any $K$) is exactly the probability that $K$ independent Bernoulli random
variables with mean $1/M$ sum to less than $MT$. We can bound the probability of this using
standard Chernoff bounds. Setting $K = 2 \cdot M^2 \cdot T$, we thus get:

$$ \Pr[W > 2 M^2 T] = \Pr[W > 2 \E[W]] \leq {\rm exp}\left( -\frac{MT}{4} \right).$$
It then suffices to show that conditioned on 
$W \le 2M^2T$, the fraction of $c$-heavy intervals (i.e., intervals where the queue contains more than $(4M)/\delta$ packets)
is bounded by $\delta$. 

In order to bound the number of $c$-heavy intervals using the bound on $W$, we first
note that $W = \sum_{t=1}^{MT} N_t$, where $N_t$ denotes the number of
packets in the queue at time $t$ (counted in microintervals). 
Furthermore, since the number of packets
in the queue can go up by at most one per microinterval, we see that heavy 
intervals contribute a lot to $W$. To make this argument precise, we
partition time into chunks containing $M/\delta$ microintervals each 
(note that ``chunks'' are much larger than the ``blocks'').
We assume here that $M/\delta$ is an integer for notational simplicity. 
For $1 \le \ell < \delta T$,
the chunk $C_\ell$ spans the range $[\ell(M/\delta), (\ell+1)(M/\delta))$. We say a chunk $C_\ell$ is
{\em bad} if the queue contains more than $(3M)/\delta$ packets at the beginning of the chunk,
and say that it is {\em good} otherwise. 
On the one hand, if a chunk is good, then every interval contained inside the chunk has at most 
$(4M)/\delta$ packets in the queue, and is hence not $c$-heavy. 
On the other hand, if a chunk $C_\ell$ is bad, then its contribution 
to $W$ (i.e., $\sum_{t \in C_\ell} N_t$) is at least $(M/\delta)(2M/\delta)$ (since
this is the minimum of $N_t$ for $t \in C_\ell$).
This allows us to show that at most a $\delta$-fraction of chunks can be bad. 
To see this, suppose $\delta_b$ is the fraction of bad chunks.
Then we have 
$$ W = \sum_{\ell =1}^{MT} N_t \geq \delta_b (\delta T) (M/\delta) (2M/\delta) = 2 (\delta_b/\delta) M^2 T.$$
Now using $W \leq 2M^2 T$, we get $\delta_b \leq \delta$.
Finally note that if $1 - \delta$ fraction of the chunks are good, then $1-\delta$
fraction of the blocks are not $c$-heavy, which
completes the proof of the lemma.
\end{proof}

\section{Conclusions}
\label{sec:conclude}
Our findings, in particular the result that 
the channel capacity is unbounded
in the setting of probabilistic
error and delay,
are surprising.
They seem to run contrary to most traditional intuition about
communication: all attempts at
reliable communication, either in the formal theory of Shannon, or in the
organic processes that led to the development of natural languages, are built
on a discrete communication model (with finite alphabet and discrete time),
even when implemented on physical (continuous time and alphabet)
communication channels. In turn such assumptions also form the basis
for our model of computing (the Turing model) and the discrete setting
is crucial to its universality. In view of the central role played by
the choice of finite alphabet in language and computation, it does make sense
to ask how much of this is imposed by nature (and the
unreliability/uncertainty it introduces) and how much due to the 
convenience/utility of the model.

Of course, our results only talk about the capacity of a certain
mathematical model of communication, and don't necessarily translate
into the physical world. The standard assumption has been that a
fixed communication channel, say a fixed copper wire, has an associated
finite limit on its ability to transmit bits (reliably). We discuss
below some of the potential reasons why this assumption may hold
and how that contrasts with our results:
\begin{description}
\item[Finite Universe] One standard working assumption in physics
is that everything in the universe is finite and discrete and the
continuous modeling is just a mathematical abstraction. While this may
well be true, this points to much (enormously) larger communication
capacities for the simple copper wire under consideration than the limits we
have gotten to. Indeed in this case, infinity would be a pretty good
abstraction also to the number of particles in the universe, and
thus of the channel capacity.
We note here that channel capacity has been studied from a purely physics
perspective and known results give bounds on the communication rate achievable
in terms of physical limits imposed by 
channel cross section, available power, Planck constant, and speed of light (see, for example,~\cite{Bekenstein'04,Lloyd+04}).

\item[Expensive Measurements] A second source of finiteness might be that
precise measurements are expensive, and so increasing the capacity does come
at increased cost. Again, this may well be so, but even if true suggests
that we could stay with existing trans-oceanic cables and keep
enhancing their
capacity by just putting better signaling/receiving instruments at
the two endpoints - a somewhat different assumption than standard ones
that would suggest the wires have to be replaced to increase capacity.
\item[Band-limited Communication] A third possibility could be that
signalling is inherently restricted to transmitting from the linear
span of a discrete
and bounded number of basis functions. As a physical assumption on nature,
this seems somewhat more complex than the assumption of probabilistic
noisiness, and, we believe, deserves further explanation/exploration.
\item[Adversaries Everywhere] Finally, there is always the possibility that
the probabilistic modelling is too weak to model even nature and we should
really consider the finite limits obtained in the adversarial setting
as the correct limits. Despite our worst-case upbringing, this does seem a
somewhat paranoid view of nature. Is there really an adversary sitting in every piece of copper wire?
\end{description}

\section*{Acknowledgments}

Thanks to Henry Cohn, Adam Kalai, Yael Kalai for helpful pointers
and discussions.

\end{document}